\documentclass[12pt,twoside]{article}
\usepackage{amsfonts,amsmath,amssymb}
\usepackage[all]{xy}
\usepackage{hyperref}
\hypersetup{colorlinks=true,linkcolor=blue,linktocpage=true,citecolor=blue}
\usepackage{amsthm}
\theoremstyle{plain}
\newtheorem{theorem}{Theorem}
\newtheorem{lemma}[theorem]{Lemma}
\newtheorem{corollary}[theorem]{Corollary}
\theoremstyle{definition}
\newtheorem{definition}[theorem]{Definition}
\newenvironment{keyword}{\centerline{\bf\small
Keywords}\begin{quote}\small}{\par\end{quote}\vskip 1ex}

\begin{document}

\title{
\vskip 2mm\bf\Large\hrule height5pt \vskip 4mm
(Non-)Equivalence of Universal Priors
\vskip 4mm \hrule height2pt}
\author{{\bf Ian Wood}$^1$ and {\bf Peter Sunehag}$^1$ and {\bf Marcus Hutter}$^{1,2}$\\[3mm]
\normalsize Research School of Computer Science \\[-0.5ex]
\normalsize $^1$Australian National University and $^2$ETH Z{\"u}rich \\[-0.5ex]
\normalsize\texttt{\{ian.wood,peter.sunehag,marcus.hutter\}@anu.edu.au}
}
\date{15 November 2011}

\maketitle              

\begin{abstract}
Ray Solomonoff invented the notion of universal induction
featuring an aptly termed ``universal'' prior probability
function over all possible computable environments
\cite{Solomonoff:64}. The essential property of this
prior was its ability to dominate all other such priors.
Later, Levin introduced another construction --- a mixture of
all possible priors or ``universal mixture''
\cite{Zvonkin:70}. These priors are well known to be
equivalent up to multiplicative constants. Here, we seek to
clarify further the relationships between these three
characterisations of a universal prior (Solomonoff's, universal
mixtures, and universally dominant priors). We see that the the
constructions of Solomonoff and Levin define an identical class
of priors, while the class of universally dominant priors is
strictly larger. We provide some characterisation of the
discrepancy.
\def\contentsname{\centering\normalsize Contents}
{\parskip=-2.7ex\tableofcontents}
\end{abstract}


\newpage
\section{Introduction}\label{sec:3-universal-priors}

In the study of universal induction, we consider an abstraction
of the world in the form of a binary string. Any sequence from
a finite set of possibilities can be expressed in this way, and
that is precisely what contemporary computers are capable of
analysing. An ``environment'' provides a measure of probability
to (possibly infinite) binary strings. Typically, the class
$\mathcal{M}$ of enumerable semimeasures is considered. Given
the equivalence between $\mathcal{M}$ and the set of monotone
Turing machines (Lemma \ref{lem:measures-TM}), this choice
reflects the expectation that the environment can be computed
by (or at least approximated by) a Turing machine.

Universal induction is an ideal Bayesian induction mechanism
assigning probabilities to possible continuations of a binary
string. In order to do this, a prior distribution, termed a
universal prior, is defined on binary strings. This prior has
the property that the Bayesian mechanism converges to the true
(generating) environment for \textit{any} environment $\mu$ in
$\mathcal{M}$, given sufficient evidence.

There are three popular ways of defining a universal prior in
the literature: Solomonoff's prior
\cite{Solomonoff:64,Zvonkin:70,Hutter:04uaibook},
\label{def:u-prior-Solomonofxf} as a universal mixture
\cite{Zvonkin:70,Hutter:04uaibook,Hutter:07uspx}, or
\label{def:u-prior-mixturex} a universally dominant semimeasure
\cite{Hutter:04uaibook,Hutter:07uspx}.
\label{def:u-prior-dominantx} Briefly, a universally dominant
semimeasure is one that dominates every other semimeasure in
$\mathcal{M}$ (Definition \ref{def:u-prior-dominant}), a
universal mixture is a mixture of all semimeasures in
$\mathcal{M}$ with non-zero coefficients (Definition
\ref{def:u-prior-mixture}), and a Solomonoff prior assigns the
probability that a (chosen) monotone universal Turing machine
outputs a string given random input (Definition
\ref{def:u-prior-Solomonof}). These and other relevant concepts
are defined in more detail in Section \ref{sec:definitions}.

Solomonoff's and the universal mixture constructions have been
known for many years and they are often used interchangeably in
textbooks and lecture notes. Their equivalence has been shown
in the sense that they dominate each other
\cite{Zvonkin:70,Hutter:04uaibook,Li:08}. We extend this result
in Section \ref{sec:UTM-is-mixture}, showing that they in fact
define exactly the same class of priors.

Further, it is trivial to see that both constructions produce
universally dominant semimeasures. The converse is, however,
not true. Universally dominant semimeasures are a larger class.
We provide a simple example to demonstrate this in Section
\ref{sec:dominant-is-not-universal}.

These results are relatively undemanding technically, however
given their fundamental nature,
that they have not to our knowledge been published to date, and
the relevance to Ray Solomonoff's famous work on universal
induction, we present them here.

The following diagram summarises these inclusion relations:

\begin{figure}
\begin{displaymath}
    \xymatrix{
&\text{Universally Dominant} \ar@/_/@{<-}[ddl]_{Lemma \ref{lem:u-mix-is-dominant}} \ar@/^/@{.x}[ddl]^{Theorem\; \ref{thm:dom-not-mixture}}&\\
&&\\
\text{Universal Mixture}\ar@/_/@{<->}[rr]^{Theorem\; \ref{thm:UTM-eq-mixture}}
  && \text{Solomonoff Prior}  \ar@/_/[uul]_{Corollary\; \ref{corol:UTM-is-dominant}}
  }
\end{displaymath}
\caption{}
\end{figure}

\section{Definitions}\label{sec:definitions}

We represent the set of finite/infinite binary strings as
$\mathbb{B}^*$ and $\mathbb{B}^\infty$ respectively. $\epsilon$
denotes the empty string, $xb$ the concatenation of strings $x$
and $b$,  $\ell(x)$ the length of a string $x$. A cylinder set,
the set of all infinite binary strings which start with some
$x\in\mathbb{B}^*$ is denoted $\Gamma_x$.

A string $x$ is said to be a prefix of a string $y$ if $y=xz$
for some string $z$. We write $x\sqsubseteq y$ or $x\sqsubset
y$ if $x$ is a proper substring of $y$ (ie: $z\ne\epsilon$). We
denote  the maximal prefix-free subset of a set of finite
strings $\mathcal{P}$ by $\lfloor \mathcal{P} \rfloor$. It can
be obtained by successively removing elements that have a
prefix in $\mathcal{P}$. The uniform measure of a set of
strings is denoted $| \mathcal{P}
|:=\sum_{p\in\lfloor\mathcal{P}\rfloor}2^{-\ell(p)}$. This is
the area of continuations of elements of $\mathcal{P}$
considered as binary decimal numbers.

There have been several definitions of monotone Turing machines
in the literature \cite{Li:08}, however we choose that which is
now widely accepted
\cite{Solomonoff:64,Zvonkin:70,Hutter:04uaibook,Li:08}
and has the useful and intuitive property Lemma \ref{lem:measures-TM}.

\begin{definition}
A monotone Turing machine is a computer with binary (one-way)
input and output tapes, a bidirectional binary work tape (with
read/write heads as appropriate) and a finite state machine to
determine its actions given input and work tape values. The
input tape is read-only, the output tape is write-only.
\end{definition}

The definitions of a universal Turing machine in the literature
are somewhat varied or unclear. Monotone universal Turing
machines are relevant here for defining the Solomonoff prior.
In the algorithmic information theory literature, most authors
are concerned with the explicit construction of a single
reference universal machine
\cite{Hutter:04uaibook,Li:08,Solomonoff:64,Turing:36,Zvonkin:70}.
A more general definition is left to a relatively vague
statement along the lines of ``a Turing machine that can
emulate any other Turing machine''. The definition below
reflects the typical construction used and is often referred to
as \textit{universal by adjunction}
\cite{Downey:10book,Figueira:06}.

\begin{definition}[Monotone Universal Turing Machine]\label{def:UTM}
A monotone universal Turing machine is a monotone Turing
machine $U$ for which there exist:
\begin{enumerate}
\item an enumeration $\{T_i:i\in\mathbb{N}\}$ of all monotone Turing machines
\item a computable uniquely decodable self-delimiting code $I:\mathbb{N}\rightarrow\mathbb{B}^*$
\end{enumerate}
such that the programs for $U$ that produce output coincide
with the set $\{I(i)p:i\in\mathbb{N},\;p\in\mathbb{B}^*\}$ of
concatenations of $I(i)$ and $p$, and
\[
  U(I(i)p) = T_i(p)\quad\forall\, i\in\mathbb{N}\;,\;p\in\mathbb{B}^*
\]
\end{definition}

A key concept in algorithmic information theory is the
assignment of probability to a string $x$ as the probability
that some monotone Turing machine produces output beginning
with $x$ given unbiased coin flip input. This approach was used
by Solomonoff to construct a universal prior
\cite{Solomonoff:64}. To better understand the properties of
such a function, we will need the concepts of enumerability
and semimeasures:

\begin{definition}
A function or number $\phi$ is said to be
\textbf{\emph{enumerable}} or \textbf{\emph{lower
semicomputable}} (these terms are synonymous) if it can be
approximated from below (pointwise) by a monotone increasing
set $\{\phi_i:i\in\mathbb{N}\}$ of finitely computable
functions/numbers, all calculable by a single Turing machine.
We write $\phi_i\nearrow\phi$. Finitely computable
functions/numbers can be computed in finite time by a Turing
machine.
\end{definition}

\begin{definition}
A \textbf{\emph{semimeasure}} is a ``defective'' probability
measure on the $\sigma$-algebra generated by cylinder sets in
$\mathbb{B}^\infty$. We write $\mu(x)$ for $x\in\mathbb{B}^*$
as shorthand for $\mu(\Gamma_x)$. A probability measure must
satisfy $\mu(\epsilon)=1$,
$\mu(x)=\sum_{b\in\mathbb{B}}\mu(xb)$. A semimeasure allows a
probability ``gap'': $\mu(\epsilon)\le1$ and
$\mu(x)\ge\sum_{b\in\mathbb{B}}\mu(xb)$. $\mathcal{M}$ denotes
the set of all enumerable semimeasures.
\end{definition}

The following definition explicates the relationship between
monotone Turing machines and enumerable semimeasures.

\begin{definition}[Solomonoff semimeasure]
\label{def:lambda_T}
For each monotone Turing machine $T$ we associate a semimeasure
\[
 \lambda_T(x) := \sum_{\lfloor p:T(p)=x*\rfloor}2^{-\ell(p)} = |T^{-1}(x*)|
\]
where $\lfloor \mathcal{P} \rfloor$ indicates the maximal
prefix-free subset of a set of finite strings $\mathcal{P}$,
$T(p)=x*$ indicates that $x$ is a prefix of (or equal to)
$T(p)$ and $\ell(p)$ is the length of $p$.
If there are no such programs, we set $\lambda_T(x):=0$. [See
\cite{Li:08} definition 4.5.4]
\end{definition}

Note that this is the probability that $T$ outputs a string
starting with $x$ given unbiased coin flip input. To see this,
consider the uniform measure given by
$\lambda(\Gamma_p):=2^{-\ell(p)}$. This is the probability of
obtaining $p$ from unbiased coin flips. $\lambda_T(x)$ is the
uniform measure of the set of programs for $T$ that produce
output starting with $x$, ie: the probability of obtaining one
of those programs from unbiased coin flips.  Note also that,
since $T$ is monotone, this set consists of a union of disjoint
cylinder sets $\{\Gamma_p:p\in\lfloor q:T(q)=x*\rfloor\}$. By
dovetailing a search for such programs and an lower
approximation of the uniform measure $\lambda$, we can see that
$\lambda_T$ is enumerable. See Definition 4.5.4 (p.299) and
Lemma 4.5.5 (p.300) in \cite{Li:08}.

An important lemma in this discussion establishes the
equivalence between the set of all monotone Turing machines and
the set $\mathcal{M}$ of all enumerable semimeasures. It is
equivalent to Theorem 4.5.2 in \cite{Li:08} (page 301) with a
small correction: $\lambda_T(\epsilon)=1$ for any $T$ by
construction, but $\mu(\epsilon)$ may not be $1$, so this case
must be excluded.

\begin{lemma}\label{lem:measures-TM}
A semimeasure $\mu$ is lower semicomputable if and only if
there is a monotone Turing machine $T$ such that
$\mu=\lambda_T$ except on $\Gamma_\epsilon \equiv
\mathbb{B}^\infty$ and $\mu(\epsilon)$ is lower semicomputable.
\end{lemma}

We are now equipped to formally define the 3 formulations for a
universal prior:

\begin{definition}[Solomonoff prior]
\label{def:u-prior-Solomonof}
The Solomonoff prior for a given universal monotone Turing machine $U$ is
\[
  M:=\lambda_U
\]
The class of all Solomonoff priors we denote $\mathcal{U}_M$.
\end{definition}

\begin{definition}[Universal mixture]\label{def:u-prior-mixture}
A universal mixture is a mixture $\xi$ with non-zero positive
weights over an enumeration $\{\nu_i:i\in\mathbb{N},
\nu_i\in\mathcal{M}\}$ of all enumerable semimeasures
$\mathcal{M}$:
\[
 \xi = \sum_{i\in\mathbb{N}}w_i\nu_i\quad:\quad \mathbb{R}\ni w_i>0\;,\;\sum_{i\in\mathbb{N}}w_i\le1
\]
We require the weights $w_{()}$ to be a lower semicomputable
function. The mixture $\xi$ is then itself an enumerable
semimeasure, i.e. $\xi\in\mathcal{M}$. The class of all
universal mixtures we denote $\mathcal{U}_\xi$.
\end{definition}

\begin{definition}[Universally dominant semimeasure]\label{def:u-prior-dominant}
A universally dominant semimeasure is an enumerable semimeasure
$\delta$ for which there exists a real number $c_\mu>0$ for
each enumerable semimeasure $\mu$ satisfying:
\[
  \delta(x) \ge c_\mu\mu(x)\quad\forall x\in\mathbb{B}^*
\]
The class of all universally dominant semimeasures we denote
$\mathcal{U}_\delta$.
\end{definition}

Dominance implies absolute continuity: Every enumerable
semimeasure is absolutely continuous with respect to a
universally dominant enumerable semimeasure. The converse
(absolute continuity implies dominance) is however not true.

\section{Equivalence between Solomonoff priors and universal mixtures}\label{sec:UTM-is-mixture}

We show here that every Solomonoff prior $M\in\mathcal{U}_M$
can be expressed as a universal mixture (i.e.:
$M\in\mathcal{U}_\xi$) and vice versa. In other words the class
of Solomonoff priors and the class of universal mixtures are
identical: $\mathcal{U}_M=\mathcal{U}_\xi$.

Previously, it was known
\cite{Zvonkin:70,Hutter:04uaibook,Li:08} that a Solomonoff
prior $M$ and a universal mixture $\xi$ are equivalent up to
multiplicative constants
\begin{align}
 M(x) &\le c_1\xi(x)   &\forall x\in\mathbb{B}^* \notag\\
 \xi(x) &\le c_2M(x) &\forall x\in\mathbb{B}^* \notag
\end{align}
The result we present is stronger, stating that the two classes
are exactly identical. Again we exclude the case $x=\epsilon$
as $M(\epsilon)$ is always one for a Solomonoff prior, but
$\xi(\epsilon)$ is never one for a universal mixture $\xi$ (as
there are $\mu\in\mathcal{M}$ with $\mu(\epsilon)<1$).

\begin{lemma} \label{thm:UTM-is-mixture}
For any monotone universal Turing machine $U$ the associated
\linebreak Solomonoff prior $M$ can be expressed as a universal
mixture. i.e. there exists an enumeration
$\{\nu_i\}_{i=1}^\infty$ of the set of enumerable semimeasures
$\mathcal{M}$ and computable function
$w_{()}:\mathbb{N}\rightarrow\mathbb{R}$ such that
\[
  M(x)=\sum_{i\in\mathbb{N}} w_i\nu_i(x)\quad\forall x\in\mathbb{B}^*\backslash\epsilon
\]
with $\sum_{i\in\mathbb{N}} w_i\le 1$ and $w_i>0\;\forall
i\in\mathbb{N}$. In other words the class of Solomonoff priors
is a subset of the class of universal mixtures:
$\mathcal{U}_M\subseteq\mathcal{U}_\xi$.
\end{lemma}
\begin{proof}
We note that all programs that produce output from $U$ are
uniquely of the form $q=I(i)p$. This allows us to split the sum
in (\ref{eqn:split-sum}) below.
\begin{align}
M(x) &=  \sum_{\lfloor q:U(q)=x*\rfloor}2^{-\ell(q)} &\notag\\
&=\sum_{i\in\mathbb{N}}\sum_{\lfloor p:U(I(i)p)=x*\rfloor}2^{-\ell(I(i)p)} &\label{eqn:split-sum} \\
&=\sum_{i\in\mathbb{N}}2^{-l(I(i))}\sum_{\lfloor p:T_i(p)=x*\rfloor}2^{-\ell(p)} & \label{eqn:UTM-is-mix-take-prefix} \notag\\
&=\sum_{i\in\mathbb{N}}2^{-l(I(i))}\lambda_{T_i}(x) &\notag 
\end{align}

Clearly $2^{-l(I(i))}>0$ and is a computable function of $i$.
Since $I$ is a self-delimiting code it must be prefix free, and
so satisfy Kraft's inequality:
\begin{equation}
   \sum_{i\in\mathbb{N}}2^{-l(I(i))} \le 1 \notag
\end{equation}

Lemma \ref{lem:measures-TM} tells us that the $\lambda_{T_i}$
cover every enumerable semimeasure if $\epsilon$ is excluded
from their domain, which shows that
$\sum_{i\in\mathbb{N}}2^{-l(I(i))}\lambda_{T_i}(x)$ is a
universal mixture. This completes the proof.
\end{proof}

\begin{corollary} \label{corol:UTM-is-dominant}
\cite{Zvonkin:70} The Solomonoff prior $M$ for a universal
monotone Turing machine $U$ is universally dominant. Thus, the
class of Solomonoff priors is a subset of the class of
universally dominant lower semicomputable semimeasures:
$\mathcal{U}_M\subseteq\mathcal{U}_\delta$.
\end{corollary}
\begin{proof}
From Lemma \ref{thm:UTM-is-mixture} we have for each
$\nu\in\mathcal{M}$ there exists $j\in\mathbb{N}$ with
$\nu=\lambda_{T_j}$ and for all $x\in\mathbb{B}^*$:
\begin{align*}
  M(x) &= \sum_{i\in\mathbb{N}}2^{-l(I(i))}\lambda_{T_i}(x) \\
       &\ge 2^{-l(I(j))}\nu(x)
\end{align*}
as required.
\end{proof}

\begin{lemma}\label{lem:u-mix-is-dominant}
Every universal mixture $\xi$ is universally dominant. Thus,
the class of universal mixtures is a subset of the class of
universally dominant lower semicomputable semimeasures:
$\mathcal{U}_\xi\subseteq\mathcal{U}_\delta$.
\end{lemma}
\begin{proof}
This follows from a similar argument to that in Corollary
\ref{corol:UTM-is-dominant}.
\end{proof}

\begin{lemma} \label{thm:mixture-is-UTM}
For every universal mixture $\xi$ there exists a universal
monotone Turing machine and associated Solomonoff prior $M$
such that
\[
\xi(x)=M(x)\quad\forall x\in\mathbb{B}^*\backslash\epsilon
\]
In other words the class of universal mixtures is a subset of
the class of Solomonoff priors:
$\mathcal{U}_\xi\subseteq\mathcal{U}_M$.
\end{lemma}

\begin{proof}
First note that by Lemma \ref{lem:measures-TM} we can find (by
dovetailing possible repetitions of some indicies) parallel
enumerations $\{\nu_i\}_{i\in\mathbb{N}}$ of $\mathcal{M}$ and
$\{T_i=\lambda_{\nu_i}\}_{i\in\mathbb{N}}$ of all monotone
Turing machines, and computable weight function $w_{()}$ with
\[
   \xi = \sum_{i\in\mathbb{N}} w_i\nu_i \quad , \quad  \sum_{i\in\mathbb{N}}w_i \le 1
\]

Take a computable index and lower approximation
$\phi(i,t)\nearrow w_i$:
\begin{align}
w_i &= \sum_t|\phi(i,t+1)-\phi(i,t)| \\
&= \sum_j 2^{-k_{ij}}\\
i,j&\mapsto k_{ij} \;\text{computable}
\end{align}
The K-C theorem
\cite{Levin:71,Schnorr:73,Chaitin:75,Downey:10book} says that
for any computable sequence of pairs $ \{k_{ij}\in\mathbb{N},\;
\tau_{ij} \in\mathbb{B}^*\}_{i,j\in\mathbb{N}}$ with $\sum
2^{-k_{ij}}\le 1$, there exists a prefix Turing machine $P$ and
strings $\{\sigma_{ij}\in\mathbb{B}^*\}$ such that
\begin{equation}
\ell(\sigma_{ij})=k_{ij}\;,\;P(\sigma_{ij})=\tau_{ij}
\end{equation}
Choosing distinct $\tau_{ij}$ and the existence of prefix
machine $P$ ensures that $\{\sigma_{ij}\}$ is prefix free. We
now define a monotone Turing machine $U$. For strings of the
form $\sigma_{ij}p$ for some $i,j$:
\begin{equation}
 U(\sigma_{ij}p) := T_i(p)
\end{equation}
For strings not of this form, $U$ produces no output. $U$
inherits monotonicity from the $T_i$, and since
$\{T_i\}_{i\in\mathbb{N}}$ enumerates all monotone Turing
machines, $U$ is universal. The Solomonoff prior associated
with $U$ is then:
\begin{align}
\lambda_U(x) &= |U^{-1}(x*)| \\
&= \sum_{i,j}2^{-\ell(\sigma_{ij})}|T_i^{-1}(x*)| \\
&= \sum_i (\sum_j2^{-k_{ij}})\lambda_{T_i}(x) \\
&= \sum_i w_i \nu_i(x) \\
&= \xi(x)
\end{align}
\end{proof}

The main theorem for this section is now trivial:

\begin{theorem} \label{thm:UTM-eq-mixture}
The classes $\mathcal{U}_M$ of Solomonoff priors and
$\mathcal{U}_\xi$ of universal mixtures are exactly equivalent.
In other words, the two constructions define exactly the same
set of priors: $\mathcal{U}_M=\mathcal{U}_\xi$.
\end{theorem}
\begin{proof}
Follows directly from Lemma \ref{thm:UTM-is-mixture} and Lemma
\ref{thm:mixture-is-UTM}.
\end{proof}

\section{Not all universally dominant enumerable semimeasures are universal mixtures}\label{sec:dominant-is-not-universal}

In this section, we see that a universal mixture must have a
``gap'' in the semimeasure inequality greater than
$c\,2^{-K(\ell(x))}$ for some constant $c>0$ independent of
$x$, and that there are universally dominant enumerable
semimeasures that fail this requirement. This shows that not
all universally dominant enumerable semimeasures are universal
mixtures.

\begin{lemma}
\label{lem:u-mix-no-gaps} For every Solomonoff prior $M$ and
associated universal monotone Turing machine $U$, there exists
a real constant $c>0$ such that
\[
  \frac{M(x)-M(x0)-M(x1)}{M(x)}\ge c\,2^{-K(\ell(x))}\quad\forall x\in\mathbb{B}^*
 \]
where the Kolmogorov complexity $K(n)$ of an integer $n$ is the
length of the shortest prefix code for $n$.
\end{lemma}
\begin{proof}
First, note that $M(x)-M(x0)-M(x1)$ measures the set of
programs $U^{-1}(x)$ for which $U$ outputs $x$ and no more.
Consider the set
\[
  \mathcal{P}:=\{ql'p\,|\,p\in \mathbb{B}^*,\,U(p)\sqsupseteq x\}
\]
where $l'$ is a shortest prefix code for $\ell(x)$ and $q$ is a
program such that $U(q{l}'p)$ executes $U(p)$ until $\ell(x)$
bits are output, then stops.

Now, for each $r=q{l}'p\in\mathcal{P}$ we have $U(r)=x$ since
$U(p)\sqsupseteq x$ and $q$ executes $U(p)$ until $\ell(x)$
bits are output. Thus $\mathcal{P}\subseteq U^{-1}(x)$ and
\begin{equation}
  |\mathcal{P}|\le |U^{-1}(x)| \label{eqn:p-u-1x}
\end{equation}
Also $\mathcal{P}=q{l}'U^{-1}(x*):=\{s=q{l}'p\,|\,p\in U^{-1}(x*)\}$, and so
\begin{equation}
  |\mathcal{P}|=2^{-\ell(q{l}')}|U^{-1}(x*)| \label{eqn:p-u-1x*}
\end{equation}
combining (\ref{eqn:p-u-1x}) and (\ref{eqn:p-u-1x*}) and noting
that $M(x)-M(x0)-M(x1)=|U^{-1}(x)|$ and $M(x)=|U^{-1}(x*)|$ we
obtain
\begin{align}
M(x)-M(x0)-M(x1) &= |U^{-1}(x)| \notag \\
&\ge |\mathcal{P}| \notag\\
&= 2^{-\ell(q{l}')}|U^{-1}(x*)| \notag \\
&= 2^{-\ell(q)}2^{-K(\ell(x))}M(x) \notag
\end{align}
Setting $c:=2^{-\ell(q)}$ this proves the result.
\end{proof}

\begin{theorem} \label{thm:dom-not-mixture}
Not all universally dominant enumerable semimeasures are
universal mixtures: $\mathcal{U}_\xi\subset\mathcal{U}_\delta$
\end{theorem}

\begin{proof}
Take some universally dominant semimeasure $\delta$, then define
$
\delta'(\epsilon):= 1,\;
\delta'(0)=
\delta'(1):=\frac{1}{2},\;
\delta'(bx):=\frac{1}{2}\delta(bx)$ for $b\in\mathbb{B}$, $x\in\mathbb{B}^*\backslash\epsilon
$. $\delta'$ is clearly a universally dominant enumerable
semimeasure with $\delta'(0)+\delta'(1)=\delta'(\epsilon)$, and
by Lemma \ref{lem:u-mix-no-gaps} it is not a universal mixture.
\end{proof}

\section{Conclusions}

One of Solomonoff's more famous contributions is the invention
of a theoretically ideal universal induction mechanism. The
universal prior used in this mechanism can be
defined/constructed in several ways.
We clarify the relationships between three different
definitions of universal priors, namely universal mixtures,
Solomonoff priors and universally dominant semimeasures. We
show that the class of universal mixtures and the class of
Solomonoff priors are exactly the same while the class of
universally dominant lower semicomputable semimeasures is a
strictly larger set.

We have identified some aspects of the discrepancy between
Solomonoff priors/universal mixtures and universally dominant
lower semicomputable semimeasures, however a clearer
understanding and characterisation would be of interest.

Since universal dominance is all that is needed to prove
convergence for universal induction
\cite{Hutter:04uaibook,Solomonoff:78} it is interesting to ask
whether the extra properties of the smaller class of Solomonoff
priors have any positive consequences for universal induction.

\subsubsection*{Acknowledgements.}
We would like to acknowledge the contribution of an anonymous
reviewer to a more elegant presentation of the proof of Lemma
\ref{thm:mixture-is-UTM}. This work was supported by ARC grant
DP0988049.

\begin{small}

\end{small}

\end{document}